%% file: Co_Le_Distributed_ToN.tex
\documentclass[11pt,onecolumn,draftcls]{IEEEtran}
\usepackage{graphicx}
\usepackage{epsfig}
\usepackage{subfigure}
\usepackage{amssymb}
\usepackage{amsbsy}
\usepackage{amsmath}
\usepackage{cite}
\usepackage{pdfpages}

\input{macros}

\begin{document}\title{Distributed Game Theoretic Optimization and Management of Multichannel ALOHA Networks}
\author{Kobi Cohen and Amir Leshem
\thanks{Kobi Cohen is with the Coordinated Science Lab, University of Illinois at Urbana-Champaign, IL 61801. Email: kobi.cohen10@gmail.com}
\thanks{Amir Leshem is with the Faculty of Engineering, Bar-Ilan University, Ramat-Gan, 52900, Israel.}
\thanks{Part of this research was presented at the International Workshop on Computational Advances in Multi-Sensor Adaptive Processing (CAMSAP), Dec. 2013.}
\thanks{This work was partially supported by ISF grant, no. 903/13.}
}
\date{}
\maketitle

%-------------abstract----------------------------------
%-------------------------------------------------------
\begin{abstract}
\label{sec:abstract}
The problem of distributed rate maximization in multi-channel ALOHA networks is considered. First, we study the problem of constrained distributed rate maximization, where user rates are subject to total transmission probability constraints. We propose a best-response algorithm, where each user updates its strategy to increase its rate according to the channel state information and the current channel utilization. We prove the convergence of the algorithm to a Nash equilibrium in both homogeneous and heterogeneous networks using the theory of potential games. The performance of the best-response dynamic is analyzed and compared to a simple transmission scheme, where users transmit over the channel with the highest collision-free utility. Then, we consider the case where users are not restricted by transmission probability constraints. Distributed rate maximization under uncertainty is considered to achieve both efficiency and fairness among users. We propose a distributed scheme where users adjust their transmission probability to maximize their rates according to the current network state, while maintaining the desired load on the channels. We show that our approach plays an important role in achieving the Nash bargaining solution among users. Sequential and parallel algorithms are proposed to achieve the target solution in a distributed manner. The efficiencies of the algorithms are demonstrated through both theoretical and simulation results.
\end{abstract}
%-------------end abstract------------------------------
%
% Keyword section
\def\keywords{\vspace{.5em}
{\bfseries\textit{Index Terms}---\,\relax%
}}
\def\endkeywords{\par}
\keywords
Collision channels, multi-channel ALOHA, best-response dynamics, Nash equilibrium, Nash bargaining solution, potential games.
%-------------section:introduction----------------------
%-------------------------------------------------------
\section{Introduction}
\label{sec:introduction}

Random access schemes have been widely used for data transmission of a large number of users sharing a common channel. In multi-channel systems, the users transmit over orthogonal channels (i.e., sub-bands) using Orthogonal Frequency Division Multiple Access (OFDMA). Each channel can be a cluster of multiple carriers. A common way to increase user rates in multi-channel systems is to exploit the channel diversity using local channel state information (CSI). Recently, multi-channel systems have been studied extensively in wireless communication \cite{Zhao_Survey_2007, Kyasanur_2009_Capacity, Leshem_Multichannel_2012, Cohen_A_Game_2012, Cohen_Game_2013, Mo_2008_Comparison}.

In this paper we examine Medium Access Control (MAC) schemes used to enable a large number of users to co-exist in a typically low number of shared channels. We investigate multi-channel ALOHA networks, where users access a channel according to a specific transmission probability. Transmission is successful if only a single user transmits over a shared channel in a given time-slot. However, if two or more users transmit simultaneously  over the same channel, a collision occurs. ALOHA-based protocols are widely used in wireless communication primarily because of their ease of implementation and their random nature. Simple transmitters can randomly access a channel without a carrier sensing operation. Past and recent works on single and multi-channel ALOHA networks can be found in \cite{Shen_Stabilized_2002, Pountourakis_Analysis_1992, Bai_Aloha_2010, Liew_2009_Bounded, Ma_2009_Analysis} and references therein. In \cite{Shen_Stabilized_2002}, stability of multi-channel networks in which a single channel is chosen randomly (from a uniform distribution) for transmission among multiple channels is studied. In \cite{Pountourakis_Analysis_1992}, a multi-channel ALOHA model, in which a single channel is used for transmissions of new packets and other channels for retransmissions, was analyzed. A Price of Anarchy (PoA) of Nash equilibria in multi-channel ALOHA networks is studied in \cite{Bai_Aloha_2010}. Queuing delay analysis for a single-channel ALOHA is provided in \cite{Liew_2009_Bounded}. Analysis of a generalized ALOHA protocol under adversarial environments is given in \cite{Ma_2009_Analysis}.

In wireless communication networks, distributed algorithms are generally preferred over centralized solutions. In this paper we mainly focus on distributed algorithms in multi-channel ALOHA networks. We examine distributed algorithms with dynamic systems where users make autonomous decisions based on local information. Such techniques have been presented in the literature. A related work on distributed optimization in cognitive radio networks can be found in \cite{Shi_Distributed_2008, Ma_Centralized_2009, Wu_Distributed_2009}. The problem of distributed learning in cognitive radio networks using multi-armed bandit technique with distributed multiple players was investigated in \cite{Liu_2010_Distributed_ICASSP}, where the number of channels is greater than the number of users and users implement carrier sensing operation before transmission. However, in this paper we adopt the ALOHA protocol for transmissions and the number of users is typically greater than the number of channels. The problem of multi-radio multi-channel allocation was investigated in \cite{Felegyhazi_Non_2007, Altman_A_Potential_2009, Vallam_Non_2011, Bai_Aloha_2010}. In \cite{Altman_A_Potential_2009}, a distributed learning algorithm was proposed that converges in some special cases. In the multi-radio multi-channel allocation model, the utility of each channel decreases with the number of radios transmitting over it. This is generally done by a TDMA protocol, for instance, among users who transmit over the same channel. As a result, users are encouraged to spread resources over channels. In this paper, however, the achievable rate of a user on a channel increases with the transmission probability (based on the ALOHA-network model) which results in strategies that allocate more resources on better channels. In \cite{Choe_2010_OFDMA}, the multi-channel ALOHA protocol in cognitive radio networks was analyzed, where the focus is on a hierarchical model of primary and secondary users in the network. The secondary users choose randomly one of the idle channels for transmission. In this paper, however, we focus on the open sharing model among users (e.g., ISM band), in which users exploit local information to choose better channels for transmissions. In \cite{Bai_Opportunistic_2006, To_2010_Exploiting}, the opportunistic multi-channel ALOHA scheme was analyzed for i.i.d Rayleigh fading channels. In this scheme, a user transmits over channels with instantaneous gains greater than some threshold. In this paper, however, long-term rates are assumed (i.e., mean-rates) and the interference caused by other users is also taken into consideration when designing effective algorithms for the spectrum access problem.

There is a significant amount of work in wireless networking that make use of game theory. Related works on networking games can be found in \cite{Altman_2006_Survey, Menache_2011_Network, Park_2012_Theory}. Random access games were studied in \cite{Mackenzie_Selfish_2001, Jin_Equilibria_2002, Meshkati_2006_Game, Menache_Rate_2008, Inaltekin_2008_Analysis, Candogan_Competitive_2009, Chen_2010_Random}. Game theoretic techniques were used in \cite{Mackenzie_Selfish_2001, Jin_Equilibria_2002, Menache_Rate_2008, Inaltekin_2008_Analysis, Candogan_Competitive_2009} to analyze single-channel ALOHA networks. In \cite{Mackenzie_Selfish_2001, Jin_Equilibria_2002, Menache_Rate_2008, Candogan_Competitive_2009}, distributed optimization algorithms of single-channel ALOHA networks using game theoretic tools are studied, where the utility of each user increases with the transmission probability. Here, we consider a similar model. Specifically, in \cite{Jin_Equilibria_2002, Menache_Rate_2008} energy-efficient Nash equilibria under user-rate demands have been established. However, the analysis of the energy-efficient equilibria does not hold under the multi-channel setting. Here, we extend this model to a multi-channel setting and study a distributed optimization of the user rates under constraints on the transmission probabilities. Another related work considered a non-cooperative power control game in multichannel networks with energy-efficiency perspectives \cite{Meshkati_2006_Game}, where the goal is to maximize the number of reliable bits transmitted per joule of energy consumed in a distributed fashion. In this paper, however, we focus on efficiency and fairness with respect to the achievable rates across users.

Cooperative game theory has been widely used to study channel sharing problems in wireless communication networks. In a non-cooperative game, players individually attempt to maximize their own utility regardless of the utility achieved by other players. On the other hand, in a cooperative game, players bargain with each other. If an agreement is reached, they act according to the agreement. If they disagree, they do not cooperate \cite{Owen_Game_1995}. An efficient solution for cooperative games is the Nash Bargaining Solution (NBS) \cite{Nash_1950_Bargaining}. In recent years, the NBS has been analyzed for the frequency flat \emph{interference} channel in the SISO \cite{Leshem_2006_Brgaining, Boche_2007_Non}, MISO \cite{Jorswieck_2008_Miso, Gao_2008_Game} and MIMO cases \cite{Nokleby_2007_Cooperative}, as well as for a frequency selective interference channel \cite{Han_2005_Fair, Leshem_2008_Cooperative, Leshem_2009_Game, Leshem_2010_Distributed, Leshem_2011_Smart}. In this paper, however, we apply cooperative game theoretic techniques to analyze the efficiency of our approach for the channel sharing problem over \emph{collision} channels in multi-channel ALOHA networks.

In our previous work \cite{Cohen_A_Game_2012, Cohen_Game_2013} we mainly focused on networks containing homogeneous users, where all users have the same transmission probability constraint. However, in this paper we focus on more general heterogenous networks, where each user in the network is allowed to transit with a different probability. Handling such cases creates additional challenges when designing effective protocols for the system. First, fairness should be considered when defining the target solutions for all users. Second, further refinements of the user dynamics are required to stabilize the system.

First, we consider the case where heterogenous users exploit their own CSI and the channel utilization to increase their utility, where each user in the network has an individual transmission probability constraint. We present the best-response algorithm that solves the distributed rate maximization. A best-response approach is a common method in non-cooperative games to achieve a Nash Equilibrium Point (NEP) \cite{Fudenberg_Game_1991, Yu_Distributed_2002, Scutari_asynchronous_2008}. The idea of best-response dynamics is that every user produces its best response in terms of the current state of all other users. Here, users need to decide which channels to access to improve their utility. The proposed best-response dynamics in this paper enable users to make autonomous decisions using their local CSI and by monitoring the load on the channels. We show that users' dynamic behavior obeys a global potential function \cite{Monderer_Potential_1996}, which implies the convergence of the dynamics.

Next, we study a simpler transmission scheme where users transmit over the channel with the highest collision-free utility (i.e., the utility that the user receives conditioned on the event that the channel is free), which is an approximate solution to the best-response dynamics as $N$ increases. The performance of the best-response dynamic are analyzed as compared to this simple transmission scheme for a finite $N$, which serves as a benchmark of the performance that could be obtained by exploiting the channel utilization. We also propose a centralized log-concave optimization problem to determine the transmission probabilities of heterogeneous users under this setting.

Finally, we consider the case where users are not restricted by a transmission probability constraint. Users are required to implement a distributed rate maximization under uncertainty since the transmission probabilities of the other users are unknown. In this case, fairness must be taken into consideration when formulating the target solution for all users. We examine the problem from a cooperative game theoretic perspective. We suggest a distributed learning scheme, where users adjust their transmission probability based on local information only to achieve the desired load on the channels to maximize their rates. We show that our approach plays an important role in achieving NBS among users. We propose sequential and parallel algorithms to reach the target solution in a distributed manner. The efficiencies of the algorithms are demonstrated through both theoretical and simulation results. Specifically, we show that the global NBS of the network can be achieved by both the sequential and parallel algorithms under mild conditions on user utilities.

The rest of this paper is organized as follows. In Section \ref{sec:network} we present the network model for the multi-channel ALOHA system. In Section \ref{sec:rate_maximization} we focus on distributed dynamics for the distributed rate maximization problem under given transmission probability constraints. In Section \ref{sec:csi} we focus on simpler solutions to rate maximization using CSI alone. In Sections \ref{sec:game} and \ref{sec:distributed} we discuss cooperative game considerations and distributed algorithms for the rate maximization problem under uncertainty of user transmission probabilities. In Section \ref{sec:simulations} we provide simulation results to demonstrate the algorithms performance.

\section{Network Model}
\label{sec:network}
We consider a wireless network containing $N$ users who transmit over $K$ orthogonal channels, where $N>K$. The users transmit over the shared channels using the slotted ALOHA protocol.
In each time slot each user is allowed to access a single channel according to a specific transmission probability. Transmission is successful if only a single user transmits over a shared channel in a given time-slot. However, if two or more users transmit simultaneously  over the same channel, a collision occurs. We assume that users are backlogged, i.e., all $N$ users always have packets to transmit.
The achievable rate of user $n$ at channel $k$ given that the channel is free, referred to as collision-free utility, is denoted by $u_n(k)\geq 0$ and is proportional to the bandwidth of channel $k$. For convenience, we define $u_n(0)=0 \;, \;\forall n$ as a virtual zero-rate channel. Transmitting over channel $k=0$ refers to no-transmission. Throughout the paper, it is assumed that the collision-free utilities are fixed during the running-time of the algorithm (i.e, $u_n(k)$ represents the mean-rate or long-term rate where the channel statistics change slowly). It is assumed that every user knows its own collision-free utility, while collision-free utilities of other users are unknown.
The collision-free rate matrix of all $N$ users in all $K+1$ channels is given by:
\beq
\label{eq:utility_matrix}
\bea{l}
\mathbf{U} \triangleq
\left[ \begin{matrix}
u_1(0) & u_1(1) & u_1(2) & \cdots & u_1(K) \\
u_2(0) & u_2(1) & u_2(2) & \cdots & u_2(K) \\
: & & & & \\
u_N(0) & u_N(1) & u_N(2) & \cdots & u_N(K)
\end{matrix} \right]  \;.
\ena
\eeq
Let $p_n(k)$ be the probability that user $n$ transmits over channel $k$.
Let $\mathcal{P}_n$ be the set of all transmission probability vectors of user $n$ in all $K+1$ channels.
A transmission probability vector $\mathbf{p}_n\in \mathcal{P}_n$ of user $n$ is given by:
\beq
\label{eq:utility_vector}
\bea{l}
\mathbf{p}_n \triangleq
\left[ \begin{matrix}
p_n(0) & p_n(1) & p_n(2) & \cdots & p_n(K)
\end{matrix} \right]  \;.
\ena
\eeq
Since we are mainly interested in high-loaded systems, where the number of users is greater (or even much greater) than the number of channels, it is desirable to limit the congestion level over the channels. Thus, we consider only single-channel strategies, where every user selects a single channel for transmission:
\beq
\label{eq:single_channel_strategy}
\bea{l}
p_n(k)=
\left\{ \begin{matrix}
1-x_n \;, & \mbox{if $k=0$}    \\
x_n \;, & \mbox{if $k=k_n$}    \\
0 \;, & \mbox{otherwise}
\end{matrix} \right.  \;,
\ena
\eeq
for some $k_n\in\left\{1, 2, ..., K\right\}$, and $0\leq x_n\leq 1$ for all $n$. We define $\mathcal{P}$ as the set of all transmission probability matrices of all $N$ users in all $K+1$ channels.
The probability matrix $\mathbf{P}\in \mathcal{P}$ is given by:
\beq
\label{eq:probability_matrix}
\bea{l}
\mathbf{P} \triangleq
\left[ \begin{matrix}
p_1(0) & p_1(1) & p_1(2) & \cdots & p_1(K) \\
p_2(0) & p_2(1) & p_2(2) & \cdots & p_2(K) \\
: & & & & \\
p_N(0) & p_N(1) & p_N(2) & \cdots & p_N(K)
\end{matrix} \right]  \;,
\ena
\eeq
where $\sum_{k=0}^{K}{p_n(k)}=1 \; \forall n$.  \\
We define $\mathcal{P}_{-n}$ as the set of all probability matrices of all $N$ users in all $K+1$ channels, except user $n$.
The probability matrix $\mathbf{P}_{-n}\in \mathcal{P}_{-n}$ is given by:
\beq
\label{eq:probability_matrix_N-1}
\bea{l}
\mathbf{P}_{-n} \triangleq
\left[ \begin{matrix}
p_1(0) & p_1(1) & p_1(2) & \cdots & p_1(K) \\
: & & & & \\
p_{n-1}(0) & p_{n-1}(1) & p_{n-1}(2) & \cdots & p_{n-1}(K) \\
p_{n+1}(0) & p_{n+1}(1) & p_{n+1}(2) & \cdots & p_{n+1}(K) \\
: & & & & \\
p_N(0) & p_N(1) & p_N(2) & \cdots & p_N(K)
\end{matrix} \right]  \;.
\ena
\eeq

When user $n$ perfectly monitors the $k^{th}$ channel utilization\footnote{Practically, the number of idle time slots and busy time slots can be used to estimate the success probability. Monitoring the channels can be done by the receiver (which can sense the spectrum from time to time and send this information to the transmitter). Another way is to monitor the null period by the transmitter as in cognitive radio systems. Any attempt to access channel $k$ by one user or more results in identifying channel $k$ as busy.}, it observes:
\beq
\label{eq:v}
\bea{l}
\displaystyle v_n(k)\triangleq \prod_{i\neq n}{\left( 1-p_i(k)\right)}=1-q_n(k)
  \;,
\ena
\eeq
which is the success probability of user $n$ on channel $k$. Roughly speaking, $q_n(k)$ can be viewed as the load that user $n$ observes on channel $k$. Increasing $q_n(k)$ decreases the rate that user $n$ can achieve over channel $k$.   \\
We further define
\beq
\label{eq:b}
\bea{l}
\displaystyle b(k)\triangleq \prod_{i=1}^{N}{\left( 1-p_i(k)\right)}
  \;,
\ena
\eeq
which is the probability that channel $k$ is not used by the users.\\
The expected rate of user $n$ in the $k^{th}$ channel is given by:
\beq
\label{eq:r}
\bea{l}
\displaystyle r_n(k)\triangleq u_n(k)v_n(k)
  \;.
\ena
\eeq
Hence, the expected rate of user $n$ is given by:
\beq
\label{eq:R}
\bea{l}
\displaystyle R_n \triangleq R_n(\mathbf{p}_n, \mathbf{P}_{-n}) = \sum_{k=1}^{K}{p_n(k)r_n(k) }
  \;.
\ena
\eeq
\section{The Distributed Rate Maximization Problem}
\label{sec:rate_maximization}
In this section we extend the results reported in \cite{Cohen_A_Game_2012, Cohen_Game_2013} for the special case of a homogenous network to the general case of a heterogenous network, where every user may have a different probability constraint. Throughout this section we consider a non-cooperative setting in the sense that every user maximizes its own rate under a constraint on the allowed transmission probability. Thus, the constraints on the attempt probabilities are used to prioritize users in the network\footnote{Similar problems for a single-channel ALOHA system were considered in \cite{Menache_Rate_2008, Jin_Equilibria_2002}, where users adjust their transmission probabilities subject to an individual rate demand. A similar approach is used in the rate-adaptive problem over interference channels in OFDM systems in which every user maximizes its own rate under a constraint on its allowed transmission power \cite{Yu_Distributed_2002})}. A question of interest under this setting is whether the system keeps oscillating due to frequent channel switching, or whether the system converges to a stable operating point (i.e., when no user can increase its rate by unilaterally switching channels). Throughout this section we addresses this question. We use the theory of potential games for purposes of convergence analysis.

We are interested in solving the distributed rate maximization problem, where each user tries to maximize its own expected rate subject to a total transmission probability constraint:
\beq
\label{eq:optimization_problem}
\bea{llll}
\displaystyle\max_{\mathbf{p}_n} \hspace{0.5cm} & R_n \hspace{0.5cm} & \mbox{s.t.} \hspace{0.5cm} & \displaystyle\sum_{k=1}^{K}{p_n(k)}\leq P_n
  \;.
\ena
\eeq
Since we are mainly interested in high-loaded systems, throughout the paper we restrict users to select at most a single channel for transmission (to reduce the collision level). Thus, $P_n<1$. Note that when user $n$ solves (\ref{eq:optimization_problem}) given the current system state, the resulting strategy is given by:
\beq
\label{eq:3_2_pure_NEP}
\bea{l}
p_n(k)=
\left\{ \begin{matrix}
1-P_n \;, & \mbox{if $k=0$}    \\
P_n \;, & \mbox{if $k=k_n^*$}    \\
0 \;, & \mbox{otherwise}
\end{matrix} \right.  \;,
\ena
\eeq
where\footnote{For the ease of presentation, we assume continuous random rates $u_n(k)$ to guarantee a uniqueness of the maximizer. Otherwise, channels with the same rate can be ordered arbitrarily.} $k_n^*=\displaystyle\arg \; \max_k\left\{ r_n(k) \right\}$, where $r_n(k)$ is defined in (\ref{eq:r}). Thus, $k_n^*$ denotes the best channel for user $n$ when its instantaneous $K$-channel utility vector is $\left[u_n(1),..., u_n(K)\right]$ and the channel utilization vector is $\left[v_n(1),..., v_n(K)\right]$.

Note that in practical systems, $\mathbf{u}_n$ is generally estimated from a pilot signal. On the other hand, complete information on matrix $\mathbf{P}_{-n}$ is not required. Knowing the channel utilization to obtain $v_n(k)$ is sufficient to make a decision.  \\
The probability matrix $\mathbf{P}$ is called the multi-strategy matrix and contains all the users' strategies, whereas $\mathbf{P}_{-n}$ is the multi-strategy matrix containing all users' strategies except the strategy of user $n$.

In the following, we define the non-cooperative multi-channel ALOHA game\footnote{This definition extends the non-cooperative multi-channel ALOHA game, defined in \cite{Cohen_A_Game_2012, Cohen_Game_2013} for homogenous users, to the general case of heterogenous users.}: \vspace{0.1cm}\\
\begin{definition}\label{def:MCA}
The non-cooperative multi-channel ALOHA (MCA) game is given by $\Gamma_{MCA}(K, P_1, P_2, ..., P_N)=\left( \mathcal{N}, \mathcal{P}, R \right)$, where $\mathcal{N}=\left\{1, 2, ..., N \right\}$ denotes the set of players (or users), $\mathcal{P}$ denotes the set of multi-strategy matrices, such that $\sum_{k=1}^{K}{p_n(k)}\leq P_n\leq 1$ for all $n\in \mathcal{N}$. $R: \mathcal{P}\rightarrow \mathbb{R}^N$, given in (\ref{eq:R}), denotes the payoff (i.e., rate) function. \vspace{0.1cm}
\end{definition}
When users cannot increase their rates by unilaterally changing their strategy, an equilibrium is obtained. \vspace{0.1cm}
\begin{definition}
A multi-strategy matrix $\mathbf{P}^*=\left[ \left(\mathbf{p}_1^*\right)^T \; \left(\mathbf{p}_2^*\right)^T \; ... \; \left(\mathbf{p}_N^*\right)^T\right]^T$ is a Nash Equilibrium Point (NEP) for the distributed rate maximization problem (\ref{eq:optimization_problem}) if
\beq
\label{eq:NEP}
\bea{l}
\displaystyle R_n(\mathbf{p}^*_n, \mathbf{P}^*_{-n}) \geq
\displaystyle R_n(\tilde{\mathbf{p}}_n, \mathbf{P}^*_{-n})  \;\; \forall n \;, \forall \tilde{\mathbf{p}}_n\in \tilde{\mathcal{P}}_n
  \;.
\ena
\eeq
where $\tilde{\mathcal{P}}_n$ denotes the set of transmission probability vectors that satisfy the constraint $\sum_{k=1}^{K}{p_n(k)}\leq P_n$.
\end{definition}

\subsection{Best-Response Dynamics}
\label{ssec:best}
Here, we propose a best-response dynamics to solve the distributed rate maximization problem. We initialize the algorithm by a simple solution where every user picks the channel with the highest collision-free utility $u_n(k)$. In the learning process step, each user monitors the channel utilization to obtain $v_n(k)$ for all $k$. Then the user updates its strategy by selecting the channel with the maximal achievable rate $r_n(k)=u_n(k)v_n(k)$ based on the estimated load.

In the best-response dynamics users can change their selected channels according to the dynamic load. In this section we show that the dynamics converge. In the following we use the theory of potential games to show that any sequential updating dynamics across users of the proposed best-response algorithm converges in finite time, starting from any point. In potential games, users' encouragement to change their strategy obeys a global potential function. Any local maximum of the potential function is a NEP of the game. In Theorem \ref{th:ordinal}, we show that $\Gamma_{MCA}(K, P_1, ..., P_N)$ is an ordinal potential game, where the utility of a player increases by unilaterally changing its strategy, if and only if the potential function increases. For the following definition $\mathcal{N}, \mathcal{P}$ are given in Definition \ref{def:MCA} and $\psi=(\psi_1, ..., \psi_N)$ is a payoff function for the $N$ users.
\\
\begin{definition}[{\cite{Monderer_Potential_1996}}]
A game $\Gamma=\left( \mathcal{N}, \mathcal{P}, \psi\right)$ is an ordinal potential game if there is an ordinal potential function $\phi : \mathcal{P}\rightarrow \mathbb{R}$ such that for every user $n\in \mathcal{N}$ and for every $\mathbf{P}_{-n}\in \mathcal{P}_{-n}$ the following holds:
\beq
\label{eq:ordinal_potential_def}
\bea{l}
\vspace{0.1cm}\displaystyle \psi_n(\mathbf{p}_n^{(2)}, \mathbf{P}_{-n})-\displaystyle \psi_n(\mathbf{p}_n^{(1)}, \mathbf{P}_{-n})>0 \iff \\ \vspace{0.1cm} \hspace{2cm}
\displaystyle \phi(\mathbf{p}_n^{(2)}, \mathbf{P}_{-n})-\displaystyle \phi(\mathbf{p}_n^{(1)}, \mathbf{P}_{-n})>0
\;, \\ \hspace{4cm}
\forall \mathbf{p}_n^{(1)}, \mathbf{p}_n^{(2)} \in \tilde{\mathcal{P}}_n \;,
\ena
\eeq
where $\tilde{\mathcal{P}}_n$ denotes the set of transmission probability vectors that satisfy the constraint $\sum_{k=1}^{K}{p_n(k)}\leq P_n$.
\end{definition}
%
%---------Theorem constrained 2-------------------
\textsl{\theorem\label{th:ordinal}{
The non-cooperative multi-channel ALOHA (MCA) game $\Gamma_{MCA}(K, P_1, ..., P_N)$ is an ordinal potential game, with the following bounded ordinal potential function:
\beq\label{eq:ordinal_potential_function}
\bea{l}
\vspace{0.1cm}
\phi(\mathbf{P}) = \displaystyle\sum_{n=1}^{N}{\sum_{k=1}^{K}{\log\left(\frac{1}{1-P_n}\right)}}\times \\ {{ \hspace{2cm}\displaystyle\left( \log\left(u_n(k)\right)-\frac{L(k)+\log\left(\frac{1}{1-P_n}\right)}{2} \right)\mathbf{1}_n(k)}}
\;,
\ena
\eeq
where
\beq\label{eq:indicator}
\mathbf{1}_n(k)=
\left\{ \begin{matrix}
1 & , & \mbox{if $p_n(k)=P_n$}  \\
0 & , & \mbox{otherwise}
\end{matrix} \right.
\eeq
is the indicator function, which indicates whether user $n$ is trying to access channel $k$, and
\beq\label{eq:L}
L(k)=\displaystyle\sum_{n=1}^{N}{\log\left(\frac{1}{1-P_n}\right)\mathbf{1}_n(k)} \;.
\eeq
}}\vspace{0.1cm} \\
%-----proof of Theorem error_exponent---
%
\begin{proof}
To prove the theorem we modify the distributed rate maximization problem (\ref{eq:optimization_problem}). Since every user selects a single channel for transmission (and $\sum_{k=1}^K p_n(k)\leq P_n$), (\ref{eq:optimization_problem}) is equivalent to the following optimization problem:
\begin{center}
$\bea{l}
\vspace{0.1cm}
\displaystyle\max_{k\in\left\{1, ..., K\right\}} \;\; R_n \;\; \mbox{s.t}\;\; p_n(k)=P_n  \;.
\ena$
\end{center}
Note that the constraint $p_n(k)=P_n$ implies $\mathbf{1}_n(k)=1$ (and also implies $p_n(0)=1-P_n$, $p_n(k')=0 \;\forall k'\neq 0, k$). As a result, for every $k$, we can multiply the objective by a constant $\mathbf{1}_n(k)\left[(1-P_n)/P_n\right]=(1-P_n)/P_n$ without affecting the solution's argument. Hence, using the monotonicity of the logarithm, (\ref{eq:optimization_problem}) is equivalent to the following optimization problem:
\beq
\bea{l}
\vspace{0.1cm}
\displaystyle\max_{k\in\left\{1, ..., K\right\}} \;\; \log\left(u_n(k)\right) -L(k) \;\; \mbox{s.t}\;\; p_n(k)=P_n\;.
\ena
\eeq
We further define:
\beq\label{eq:R_n_hat}
\bea{l}
\displaystyle\psi_n(\mathbf{p}_n, \mathbf{P}_{-n})=\psi_n(k, \mathbf{P}_{-n})\triangleq\tilde{u}_n(k)-L(k) \;,
\ena
\eeq
where $\mathbf{p}_n$ is determined by the chosen channel $k$ and $\tilde{u}_n(k)=\log\left(u_n(k)\right)$.\\
Next, assume that user $n_0$ selects channel $k_1$ according to strategy $\mathbf{p}_{n_0}^{(1)}$ and changes its strategy by selecting channel $k_2$ according to strategy $\mathbf{p}_{n_0}^{(2)}$. In what follows $\mathbf{1}_n^{(i)}(k), L^{(i)}(k)$ refer to $\mathbf{1}_n(k), L(k)$ with respect to strategy $\mathbf{p}_n^{(i)}$, for $i=1, 2$.
The difference in the payoff function $\Delta\psi_{n_0}$ is given by:
\begin{center}
$\bea{l}
\vspace{0.1cm}
\Delta\psi_{n_0}=\displaystyle \psi_{n_0}(\mathbf{p}_{n_0}^{(2)}, \mathbf{P}_{-n_0})-\displaystyle \psi_{n_0}(\mathbf{p}_{n_0}^{(1)}, \mathbf{P}_{-n_0}) \\ %\hspace{0.9cm}
=\left[\tilde{u}_{n_0}(k_2)-L^{(2)}(k_2)\right] - \left[\tilde{u}_{n_0}(k_1)-L^{(1)}(k_1)\right]. \vspace{0.1cm}
\ena$
\end{center}
We apply the ordinal potential function that was introduced in \cite{Mavronicolas_Congestion_2007} to show that the difference in the proposed function (\ref{eq:ordinal_potential_function}) $\Delta\phi$ is given by:
\begin{center}
$\bea{l}
\vspace{0.1cm}
\Delta\phi=\displaystyle \phi(\mathbf{p}_{n_0}^{(2)}, \mathbf{P}_{-n_0})-\displaystyle \phi(\mathbf{p}_{n_0}^{(1)}, \mathbf{P}_{-n_0}) \\ \hspace{0.0cm}
\displaystyle=\sum_{n=1}^{N}{\sum_{k=1}^{K}{\tilde{p}_n}%\times \\
{ \displaystyle\left( \tilde{u}_n(k)-\frac{L^{(2)}(k)+\tilde{p}_n}{2} \right)\mathbf{1}^{(2)}_n(k)}} \\ \hspace{1cm}
-\displaystyle\sum_{n=1}^{N}{\sum_{k=1}^{K}{\tilde{p}_n}}%\times \\
{{ \displaystyle\left( \tilde{u}_n(k)-\frac{L^{(1)}(k)+\tilde{p}_n}{2} \right)\mathbf{1}^{(1)}_n(k)}} \\
=\displaystyle \tilde{p}_{n_0} \left(\left[\tilde{u}_{n_0}(k_2)-L^{(2)}(k_2)\right] \right. \vspace{0.1cm}\\ \hspace{3cm} \left.
- \displaystyle\left[\tilde{u}_{n_0}(k_1)-L^{(1)}(k_1)\right]\right)\vspace{0.1cm}\\ \hspace{1cm}
=\displaystyle\log\left(\frac{1}{1-P_{n_0}}\right)\Delta\psi_{n_0},
\ena$
\end{center}
where $\tilde{p}_n=\log\left(1/(1-P_n)\right)$. \\
Hence, (\ref{eq:ordinal_potential_def}) follows.
Furthermore, $\phi(\mathbf{P})$ is upper bounded by $\phi(\mathbf{P})<\sum_{n=1}^{N}{\max_{k}{\log\left(\frac{1}{1-P_n}\right)\log\left({u}_n(k)\right)}}$. \\
Due to the monotonicity of the logarithm increasing $\psi_n$ increases the actual rate $R_n$. As a result, $\phi(\mathbf{P})$ (\ref{eq:ordinal_potential_function}) is a bounded ordinal potential function of $\Gamma_{MCA}(K, P_1, ..., P_N)$ which completes the proof.
\end{proof}
%%--------------------------------------------------
%
\hspace{0cm}\vspace{0.1cm}
\begin{corollary}
The proposed sequential best-response algorithm converges to a NEP in finite time, starting from any point.
\end{corollary}

\section{Competitive Approach Under the Totally Greedy (TG) Access Algorithm}
\label{sec:csi}
In this section we focus on the simple transmission scheme where users access the channel with the highest collision-free utility, without considering the channel utilization. The users have constraints on the transmission probability, as in the previous section. We refer to this scheme as the Totally Greedy (TG) access scheme. The disadvantage of this scheme is that users do not exploit the channel load information to increase their rate. For instance, consider the case of two channels $k_1$, $k_2$. Assume that an interferer exists on channel $k_2$; thus all users observe $u_n(k_1)>u_n(k_2)$. Using the TG scheme, all users transmit over channel $k_1$ even if the load on this channel is significantly higher than the load on channel $k_2$. This scheme may lead to inefficient exploitation of the spectrum band. On the other hand, it is simple to implement and only a single iteration is required. Furthermore, under some mild conditions on the utility matrix it provides a good solution as the number of users increases (as will be discussed in subsequent sections). Thus, it can serve as a benchmark of the performance that could be obtained by exploiting the channel utilization when implementing the best-response dynamics. In Section \ref{ssec:BC_Vs_BR} we examine the system performance in terms of user sum rate, when users exploit the channel utilization to improve their rates in a distributed fashion as compared to the TG scheme.

Let $\widetilde{U}$ be the actual utility matrix, which is obtained by removing the first column (i.e., the all-zero vector) from $U$, defined in (\ref{eq:utility_matrix}). For purposes of analysis in this section we assume some weak conditions on the utility matrix $U$:
\begin{description}
  \item[A(1)] The rows in the matrix $\widetilde{U}$ are statistically independent.
  \item[A(2)] The columns in the matrix $\widetilde{U}$ are identically distributed.
\end{description}
Due to path loss attenuation, the rows in the matrix $\widetilde{U}$ (which refer to users) are assumed to be independent but not-necessarily identically distributed.
Due to the frequency selective fading effect, the columns for each row in the matrix $\widetilde{U}$ (which refer to channels) are assumed to be identically distributed but not-necessarily independent.
It was shown in \cite{Cohen_Game_2013} that when assumptions $A(1), A(2)$ hold, the TG scheme provides an approximate solution to the best-response dynamics discussed in the previous section as $N$ increases.
The intuition for this result is that for a large number of users, the number of users that select channel $k$ approaches $N/K$. Hence, the load approaches a constant value and selecting the channel with the highest collision-free utility is more dominant. Furthermore, setting $P_n=K/N$ maximizes the network throughput since the expected number of users that select channel $k$ is $N/K$.

\subsection{Totally Greedy Vs. Best Response}
\label{ssec:BC_Vs_BR}
Here, we examine the loss of the simple TG scheme as compared to the best-response dynamics for a finite number of users in the case where every user experiences equal rates for all channels, i.e., $u_n=u_n(k)=u_n(k')$ for all $k, k'$. We consider the case where all users set $P_n=K/N$ to maximize the network throughput in terms of sum rate \cite{Cohen_Game_2013, Bai_Opportunistic_2006}. In this case the TG scheme randomly picks a channel.

Let $N(k)$ be the number of users that select channel $k$ and assume that $N/K \in \mathbb{Z}$. Then, the best-response dynamics converge when $N(k)=N/K$ for all $k$. The achievable rate of user $n$ is given by:
\beq
R_n^{BR}=\displaystyle u_n\frac{K}{N}\cdot \left(1-\frac{K}{N}\right)^{\frac{N}{K}-1}\;.
\eeq
Hence, the sum rate achieved by the best-response dynamics is given by:
\beq
S_R^{BR}=\displaystyle K\cdot \left(1-\frac{K}{N}\right)^{\frac{N}{K}-1}\frac{1}{N}\sum_{n=1}^{N}{u_n} .
\eeq
Next, we compute the expected user sum rate achieved by the TG scheme. Assume that user $n$ transmits over channel $k$. Note that channel $k$ is selected by all other users with a probability $1/K$ and then every user that picks channel $k$ actually transmits over it with a probability $K/N$. Therefore, the expected rate of user $n$ on channel $k$ is: $R_n(k)=u_n\frac{K}{N}\left(1-\frac{1}{K}\cdot\frac{K}{N} \right)^{N-1}$. Since every channel is selected with an equal probability $1/K$, the expected rate of user $n$ achieved by the TG scheme is given by:
\beq
R_n^{TG}=\displaystyle u_n\frac{K}{N}\cdot \left(1-\frac{1}{N}\right)^{N-1} \;.
\eeq
Hence, the expected sum rate achieved by the TG scheme is given by:
\beq
S_R^{TG}=\displaystyle K\cdot \left(1-\frac{1}{N}\right)^{N-1}\frac{1}{N}\sum_{n=1}^{N}{u_n}.
\eeq
Note that the sum rate achieved by both schemes approaches $Ke^{-1}\frac{1}{N}\sum_{n=1}^{N}{u_n}$ as $N$ increases. \\
The gain of the best response algorithm over the TG scheme is defined as the ratio between the sum rate achieved by the best-response dynamics and the sum rate achieved by the TG scheme. The gain is given by:
\beq
\displaystyle\rho=\frac{S_R^{BR}}{S_R^{TG}}= \frac{\left(1-\frac{K}{N}\right)^{\frac{N}{K}-1}}{\left(1-\frac{1}{N}\right)^{N-1}}.
\eeq
It can be shown that $\rho >1$ and that $\rho\rightarrow 1$ as $N/K$ increases. The intuition for this result is that as $N/K$ increases, the number of users that select channel $k$ approaches $N/K$. Hence, the load approaches a constant value and the TG selection is more dominant.
To illustrate the result, we depict $\rho$ in Fig. \ref{fig:rho}. It can be seen that the best-response algorithm outperforms the TG scheme by roughly $260\%$ when $N/K=1$ and by $20\%$ when $N/K=3$.
\begin{figure}[h]
\centering \epsfig{file=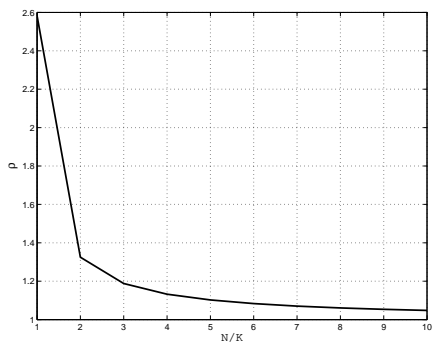,
width=0.35\textwidth}
\caption{The gain of the best-response dynamics over the TG scheme as a function of $N/K$.}
\label{fig:rho}
\end{figure}

\subsection{Determining $P_n$ for Heterogenous Networks}
\label{ssec:determining}
In this section we discuss the choice of $P_n$, $n=1, ..., N$. Assume that $P_n=\alpha_n/N$. A natural criterion for rate maximization in communication networks is to maximize the rate of a specific user (say user $1$) subject to the target rate constraints of all other users \cite{Cendrillon_Optimal_2006}. Note that as long as the demands for users $n=2, 3, ..., N$ are inside the rate region (i.e., feasible demands), maximizing the rate of user $1$ brings the system to operate on the boundary of the rate region, which is a desired operating point. %
We assume that $A(1), A(2)$ hold. Let $k_n^*=\arg\max_k u_n(k)$. Since we assume identically distributed channels, the probability that $k_n^*=k$ is $1/K$ for all $k$ and for all $n$, and the probability. Hence, the expected rate of user $n$ is given by:
\beq\label{eq:expected_rates}
\bea{l}
\displaystyle \E\left\{R_n\right\}=\E\left\{u_n(k_n^*)\right\}\frac{\alpha_n}{N} \prod_{i\neq n}{\left( 1-\frac{\alpha_i}{NK}\right)} \;.
\ena
\eeq
We consider the problem of maximizing the rate of a specific user $\E\left\{u_1(k_1^*)\right\}\frac{\alpha_1}{N} \prod_{i\neq 1}{\left( 1-\frac{\alpha_i}{NK}\right)}$ such that all other user rates satisfy the target rate demands, $\E\left\{u_n(k_n^*)\right\}\frac{\alpha_n}{N} \prod_{i\neq n}{\left( 1-\frac{\alpha_i}{NK}\right)}\geq R_n^T$ for all $n\neq 1$. Let $\bar{R}_n^T=R_n^T/\E\left\{u_n(k_n^*)\right\}$. Since $\E\left\{u_1(k_1^*)\right\}$ is a constant independent of $\alpha_1, ..., \alpha_N$, we need to solve the following optimization problem: \\
\beq
\label{eq:optimization__alpha1_2}
\bea{ll}
\displaystyle\arg\;\max_{\alpha_1, ..., \alpha_N} \hspace{0.5cm} & \displaystyle\frac{\alpha_1}{N} \prod_{i\neq 1}{\left( 1-\frac{\alpha_i}{NK}\right)} \hspace{0.5cm} \vspace{0.3cm}\\\hspace{0.5cm}
 \hspace{0.7cm}\mbox{s.t.} \hspace{0.5cm} & \displaystyle\frac{\alpha_n}{N} \prod_{i\neq n}{\left( 1-\frac{\alpha_i}{NK}\right)}\geq \bar{R}_n^T \hspace{0.5cm} \forall n\neq 1\vspace{0.3cm}
  \;.
\ena
\eeq
We optimize over $\alpha_1, ..., \alpha_N$ to maximize user $1$'s expected rate, such that target rate demands for all other users are satisfied. \\
The optimization problem (\ref{eq:optimization__alpha1_2}) is log-concave. Complexity does not depend on the number of channels $K$. Note
that reducing $\alpha_n$ increases all the other user rates $i\neq n$. Hence, the optimal solution lies on the boundary of the rate constraints.

\section{Cooperative Game Theoretic Learning}
\label{sec:game}
In previous sections we examined the dynamics of multi-channel ALOHA networks, when users try to maximize their rates under given transmission probability constraints. In this section we consider a different problem in multi-channel ALOHA networks, where the transmission probability constraints are not given. As a result, a self control on the transmission probability is mandatory to avoid high load on the channels and consequently a significant loss in data rate.

Unlike the homogenous users scenario, here we do not consider the sum rate as a performance measure of the network due to fairness considerations. Note that the optimal solution for the the sum rate maximization is when a single user with the highest collision-free utility on every channel transmits with probability $1$, while all the other users do not transmit. This operating point is clearly very bad from a fairness perspective. Therefore, in this section the sum log rate is considered to be a performance measure of the network, which is a common measure to evaluate the tradeoff between efficiency and fairness among users \cite{Yu_2011_Multicell, Gai_2011_Packet}. We show that our approach plays an important role in achieving NBS among users \cite{Nash_1950_Bargaining, Owen_Game_1995}.

First, in Section \ref{ssec:motivation} we motivate our approach by analyzing the performance among users that transmit over the same channel. Roughly speaking, we show that $b(k)\rightarrow e^{-1}$ is essential to achieve both efficiency and fairness among users that transmit over channel $k$. Based on this observation, we formulate the distributed rate maximization for a multi-channel network in Section \ref{ssec:optimization_adaptive}. In Section \ref{ssec:performance} we analyze the performance for the entire network. We show that when assumptions $A(1), A(2)$ hold, our approach achieves the target solution among all the users in the multi-channel network (and not just for each channel separately).

\subsection{Rationale}
\label{ssec:motivation}

Let $\mathcal{N}_k\;, \; N(k)=|\mathcal{N}_k|$ be the set of users that transmit over channel $k$ and its cardinality, respectively. In this section, we show that $b(k)\rightarrow e^{-1}$ is essential to achieving both efficiency and fairness among users in $\mathcal{N}_k$.

\subsubsection{Fairness in Channel Sharing}
\label{sssec:fairness}

\begin{definition}
A transmission scheme is called an \emph{equal share} transmission scheme over channel $k$ if $P_n=p_k$ for some $0\leq p_k\leq 1$ for all $n\in\mathcal{N}_k$. \vspace{0.05cm}
\end{definition}
Applying the equal share transmission scheme is reasonable from a fairness perspective, where users that transmit over the same channel are required to equally share the expected number of successful time slots. Thus, in Proposition \ref{prop:b} we consider the case where users that transmit over channel $k$ are restricted to using the equal share transmission scheme. It is shown that $b(k)\rightarrow e^{-1}$ is a necessary condition to maximize the user rates under this setting as the number of users increases.
%------------------proposition----------------
\begin{proposition}
\label{prop:b}
Assume that the equal share transmission scheme over channel $k$ is implemented. Then, setting $P_n=1/N(k)$ for all $n\in \mathcal{N}_k$ maximizes the user rate $R_n$ for all $n\in\mathcal{N}_k$. \vspace{0.1cm}
\end{proposition}
Proposition \ref{prop:b} follows from standard results on a single-channel ALOHA network \cite{Bertsekas_1992_Data}. \vspace{0.1cm}
%----------corollary-----------------------
\begin{corollary}
\label{cor:b1}
Maximizing the user rate $R_n$ for all $n\in\mathcal{N}_k$ under the equal share transmission scheme implies $b(k)\rightarrow e^{-1}$ as $N(k)\rightarrow\infty$. \vspace{0.1cm}
\end{corollary}
%
%----------proof---------------------
\begin{proof}
Setting $P_n=1/N(k)$ for all $n\in\mathcal{N}_k$ yields:
\beq
\label{eq:b_asymp}
\displaystyle b(k)=\left(1-\frac{1}{N(k)}\right)^{N(k)}\rightarrow e^{-1} \;\;\mbox{as} \;\; N(k)\rightarrow\infty
\eeq
\end{proof}

\subsubsection{The Efficiency and Fairness Tradeoff}
\label{sssec:balancing}
\hspace{0cm}\\
Next, to further strengthen the rationale, we examine the case when the transmission probability may be different for every user and users may transmit with a probability close to $1$. Note that the sum rate is maximized by setting $P_{n'}=1$ for $n'=\arg\;\max_{n\in\mathcal{N}_k}\left(u_n(k)\right)$, and $P_n=0$ for all $n\neq n'$, which obviously does not maintain fairness. On the other hand, Theorem \ref{th:b2} shows that the equal share transmission scheme still maximizes the user sum log rate over channel $k$ (i.e., the tradeoff between efficiency and fairness among users that share channel $k$ is good).  \vspace{0.1cm}
%
%------theorem--------------
\begin{theorem}
\label{th:b2}
The unique solution that maximizes the sum log rate over channel $k$, $\arg\;\max_{\left\{P_n\right\}_{n\in\mathcal{N}_k}}\;\sum_{n\in\mathcal{N}_k}\log(R_n)$ is given by $P_n^*=1/N(k)$ for all $n\in\mathcal{N}_k$. \vspace{0.1cm}
\end{theorem}
%
%------proof----------------
\begin{proof}
The achievable rate of user $n$ is given by:
\beq
R_n=u_n(k)P_n\prod_{i\in\mathcal{N}_k \;,\; i\neq n}\left(1-P_i\right)\;, \forall n\in\mathcal{N}_k \;.
\eeq
Taking log on both sides yields:
\beq
\bea{l}
\displaystyle\log(R_n)=\log(u_n(k))+\log(P_n)+\sum_{i\in\mathcal{N}_k \;,\; i\neq n}\log\left(1-P_i\right) \vspace{0.1cm}\\\hspace{5cm}
\;, \forall n\in\mathcal{N}_k \;.
\ena
\eeq
Let $S_k\triangleq\sum_{n\in\mathcal{N}_k}\log(R_n)$ be the sum log rate on channel $k$. Hence, for $N(k)\geq 2$ we obtain:
\beq
\label{eq:S_k}
\bea{l}
\displaystyle S_k=\sum_{n\in\mathcal{N}_k}\left[\log(u_n(k))+\log(P_n) \vspace{0.1cm} \right.\\\left. \hspace{3cm}
                \displaystyle+\sum_{i\in\mathcal{N}_k \;,\; i\neq n}\log\left(1-P_i\right)\right] \vspace{0.1cm}\\
\displaystyle =\sum_{n\in\mathcal{N}_k}\left[\log(u_n(k))\right] \vspace{0.1cm} \\ \hspace{0.5cm}
\displaystyle+\sum_{n\in\mathcal{N}_k}\left[\log(P_n)+\sum_{i\in\mathcal{N}_k}\log\left(1-P_i\right)
                                                                -\log\left(1-P_n\right)\right] \vspace{0.1cm}\\
\displaystyle =\sum_{n\in\mathcal{N}_k}\left[\log(u_n(k))\right] \vspace{0.1cm} \\ \hspace{0.5cm}
\displaystyle +\sum_{n\in\mathcal{N}_k}\log(P_n)+\left(N(k)-1\right)\sum_{n\in\mathcal{N}_k}\log\left(1-P_n\right)
\;,
\ena
\eeq
and for $N(k)=1$ we have:
\begin{center}
$\displaystyle S_k=\sum_{n\in\mathcal{N}_k}\left[\log(u_n(k))+\log(P_n)\right]$\;.
\end{center}
By the monotonicity of the logarithm, it is clear that for $N(k)=1$, maximizing $S_k$ yields $P_n^*=1/N(k)=1$ for all $n\in\mathcal{N}_k$. Next, we focus on the case where $N(k)\geq 2$. Note that $S_k$ is a strictly concave function of $P_n, n\in\mathcal{N}_k$. Therefore, it has a unique global maximum. Differentiating $S_k$ with respect to $P_n\;, \; n\in\mathcal{N}_k$, and equating to zero yields a unique solution $P_n^*=1/N(k)$ for all $n\in\mathcal{N}_k$.
\end{proof}

As a result, we obtain the following corollary, as was done in (\ref{eq:b_asymp}).\vspace{0.1cm}
\begin{corollary}
\label{cor:b2}
Maximizing the sum log rate on channel $k$ implies $b(k)\rightarrow e^{-1}$ as $N(k)\rightarrow\infty$. \vspace{0.1cm}
\end{corollary}

\subsubsection{Bargaining Over the Collision Channel}
\label{sssec:bargaining}

\hspace{0cm}\\
Here, we provide an interpretation of our approach from a cooperative game theory perspective. In a non-cooperative game, players (i.e., users) individually attempt to maximize their own utility regardless of the utility achieved by other players. On the other hand, in a cooperative game, players bargain with each other. If an agreement is reached, they act according to the agreement. If they disagree, they do not cooperate. For more details on cooperative game theory and applications to network games, the reader is referred to \cite{Nash_1950_Bargaining, Owen_Game_1995, Han_2005_Fair, Leshem_2006_Brgaining, Boche_2007_Non, Nokleby_2007_Cooperative, Jorswieck_2008_Miso, Gao_2008_Game, Leshem_2008_Cooperative, Leshem_2009_Game}.

Let $\mathcal{N'}$ be the set of $N'=|\mathcal{N}'|$ players. The underlying structure for Nash bargaining in an $N'$ players scenario is a set of outcomes of the bargaining process $R_c\in\mathbb{R}^{N'}$ (which in our model represents the set of achievable rates that the users can get by cooperating) and a designated disagreement outcome $\mathbf{d}=(R_1^{min}, ..., R_{N'}^{min})$ (where in our model $R_n^{min}$ represents the minimal rate that user $n$ would expect to achieve. Otherwise, it will not cooperate). Cooperative game theories prove that there exists a unique and efficient solution under intuitive axioms of fairness, symmetry and scaling-invariant and this solution is given by \cite{Owen_Game_1995}:
\beq
\label{eq:NBS}
R^*=\arg\;\max_{R\in R_c\cup \left\{\mathbf{d}\right\}}\prod_{n\in\mathcal{N'}}\left(R_n-R_n^{min}\right) \;,
\eeq
dubbed the Nash Bargaining Solution (NBS) among players in $\mathcal{N'}$.

Next, we show that maximizing the sum log rate (i.e., applying the equal share transmission scheme) over channel $k$ is also an NBS among users in $\mathcal{N}_k$. \vspace{0.1cm}
%------------------theorem-------------------------
\begin{theorem}
\label{th:SL_SN}
Let $\mathcal{N'}=\mathcal{N}_k$ in (\ref{eq:NBS}). Setting $P_n=1/N(k)$ for all $n\in\mathcal{N}_k$ achieves the NBS among users in $\mathcal{N}_k$.
\end{theorem}
%----------proof------------
\begin{proof}
Note that by non-cooperating all the users in $\mathcal{N}_k$ will increase their transmission probabilities to $1$ to increase their rates. Thus, every user in $\mathcal{N}_k$ (say $n$) expects to obtain $R_n^{min}=0$ by non-cooperating. Thus, substituting $R_n^{min}=0$ for all $n\in\mathcal{N}_k$ in (\ref{eq:NBS}) yields the sum log rate maximization. The rest of the proof follows from the proof of Theorem \ref{th:b2}.
\end{proof}
\vspace{0.1cm}
\begin{corollary}
\label{cor:b3}
Applying the NBS among users that share channel $k$ implies $b(k)\rightarrow e^{-1}$ as $N(k)\rightarrow\infty$. \vspace{0.1cm}
\end{corollary}

In Section \ref{ssec:performance} we show that when assumptions $A(1), A(2)$ holds, our approach achieves the global NBS of the network.

\subsection{The Optimization Problem}
\label{ssec:optimization_adaptive}

In this section we formulate the distributed rate maximization for a multi-channel network aimed to achieve both efficiency and fairness on every channel. In subsequent sections we examine two schemes used to solve the proposed optimization problem in a distributed fashion. Moreover, in Section \ref{ssec:performance} we show that when $A(1), A(2)$ hold, not just the sum log rate on every channel is maximized, but also the global sum log rate of the network $\sum_{n=1}^{N}\log(R_n)$ is maximized as $N$ increases (which is also the global NBS of the network as shown in Theorem \ref{th:global_NBS}).

Based on the observation that for a large number of users $b(k)$ should approach $e^{-1}$, the goal in this section is to cause the system to operate with the desired load on each channel in a distributed fashion. Let $b_n(k)$ be the estimate of $b(k)$ at user $n$ by monitoring the channel utilizations. Hence, each user is required to maximize its rate, but maintain a desired load on the channels (which is affected by $b_n(k)$):
\beq
\label{eq:optimization_uncertainty}
\bea{llll}
\displaystyle\max_{k\in\left\{1, ..., K\right\}} \hspace{0.5cm} & R_n \hspace{0.5cm} & \mbox{s.t.} \hspace{0.5cm} & b_n(k)=e^{-1}
  \;.
\ena
\eeq
We refer to this formulation as the \emph{adaptive rate maximization} problem, since the transmission probabilities are adapted to the channel loads. \\
Note that solving this problem may lead to undesirable solutions depending on the dynamic updating of the transmission probabilities across users (note that $b(k)\rightarrow e^{-1}$ is a necessary but not a sufficient condition to maximize the sum log rate). For instance, assume that user $n$ monitors $v_n(k)$ and wants to force its transmission probability to satisfy the constraint: $b_n(k)=(1-P_n)\cdot v_n(k)=e^{-1}$. In this case, the update of $P_n$ yields
\beq
\displaystyle P_n=\max{\left\{1-\frac{e^{-1}}{v_n(k)}\;,\;0\right\}} \;.
\eeq
As a result, if user $n$ detects channel $k$ as a free channel, i.e., $v_n(k)=1$, it maximizes its probability to get $P_n=1-e^{-1}$ which satisfies the constraint. Then, in the next iteration, any other user that accesses this channel will detect $v_n(k)=e^{-1}$ and will force its probability to zero to satisfy the constraint (as a result, $\sum_{n}\log(R_n)\rightarrow-\infty$). Hence, in the next section we propose two schemes to obtain the target solutions for all users.

\section{Distributed Algorithms for the Adaptive Rate Maximization Problem}
\label{sec:distributed}

In this section we propose parallel and sequential mechanisms to solve (\ref{eq:optimization_uncertainty}) efficiently. The proposed mechanisms are executed from time to time until convergence. It should be noted that the proposed algorithms apply for all $N\geq 1$ and perform well as can be seen via simulation results. Performance analysis, however, will be presented under the asymptotic regime (i.e., as $N$ approaches infinity) and an accurate estimate of $b_n(k)$.

\subsection{Sequential Updating}
\label{ssec:sequential}
In the sequential updating mechanism, users adjust their transmission probability until they get the desired channel load. Let $\delta_n(k)\triangleq|b_n(k)-e^{-1}|$. The users' goal is to reduce $\delta_n(k)$ sequentially until convergence. \\
In the initialization step, all users select the channels with the highest collision-free utility and set their transmission probability to $P_n=p^{(0)}_n(k^{*})=p_0 << 1$. \\
Next, in the learning step, each user occasionally monitors the channel utilization $v_n(k)$ of all channels. After the user has estimated $v_n(k)$ it does the following. First, it computes the highest transmission probability allowed on each channel based on the estimated load:
\beq
\displaystyle \tilde{p}_n(k)=\max{\left\{1-\frac{e^{-1}}{v_n(k)}\;,\;0\right\}} %\;\;, \;\;  \forall k\neq k^{*}
\;.
\eeq
This operation will encourage users to move to channels with low loads. \\
Next, the user computes the potential achievable rates on all the channels:
\beq
\tilde{R}_n(k)=\tilde{p}_n(k) u_n(k)v_n(k) \;.
\eeq
If there is a channel with a higher potential rate than its current channel, the user switches to this channel; i.e., it updates $k^{*}$ as follows:
\beq
\displaystyle k^{*}=\arg\; \max_{k}{\left\{\tilde{R}_n(k)\right\}} \;.
\eeq
Next, the user reduces $\delta_n(k^{*})$ to obtain the desired load. If $b_n(k^{*})=(1-P_n)\cdot v_n(k^{*})>e^{-1}$, user $n$ increases its transmission probability to increase the load on the channel: $P^{(\ell)}_n=P^{(\ell-1)}_n+\epsilon$. Otherwise, it reduces its transmission probability to reduce the load on the channel: $P^{(\ell)}_n=P^{(\ell-1)}_n-\epsilon$. \\
Note that as $b_n(k)$ approaches $e^{-1}$ for all $k$, the potential transmission probability $\tilde{p}_n(k)$ that user $n$ computes for all other channels $k\neq k^{*}$ approaches zero to maintain the desired load. Hence, users are encouraged to remain in their channels as the load approaches the desired load.  \\
To stabilize the algorithm, we allow user $n$ to switch to channel $k_2$ from $k_1$ only if it gains at least $\delta_{R}(n)$ percents of its current rate: $R_n(k_2)\geq R_n(k_1)\left(1+\delta_R(n)\right)$. Users may update $\delta_R(n)$ dynamically to speed up convergence (i.e., by increasing $\delta_R(n)$) or to increase their data rate (i.e., by reducing $\delta_R(n)$) from time to time\footnote{Practically, simulation results show convergence of the sequential updating algorithm for very small values of $\delta_R(n)$.}.
The algorithm stops when $b_n(k)\approx e^{-1}$ for all $k$. The sequential updating mechanism is given in Table \ref{tab:seq_algorithm}. For $\delta_R(n)=0$ users play their best response, while for $\delta_R(n)\rightarrow\infty$ users select the channel with the highest collision-free utility.
 \vspace{0.1cm}  \\

\begin{remark}
Note that setting $\delta_R(n)\rightarrow \infty$ leads to the simple TG scheme discussed in Section \ref{sec:csi}. If A(1), A(2) hold, the TG scheme performs well for a large number of users and $\delta_R(n)\rightarrow \infty$ is a good choice.
On the other hand, setting $\delta_R(n)\rightarrow \infty$ may lead to undesirable solutions in non-i.i.d utility matrix scenarios. For instance, consider the case of $K=2$ channels, where all users detect $u_n(1)>u_n(2)$ for all $n$. This case is commonplace in communication networks when there is significant interference on channel $k=2$. In this case, by using the TG scheme, all users transmit over channel $k$, which may cause a very high load on this channel.
\end{remark}
\begin{table}
\caption{Sequential updating algorithm}
\centering
\normalsize\begin{tabular}{|l|}
\hline
\\
\hspace{0.3cm} - \hspace{0.3cm} Initialize:\\
\hspace{0.3cm} - \hspace{0.3cm} for $n=1, ..., N$ users do:\\
\hspace{0.3cm} - \hspace{1cm} estimate $u_n(k)$ for all $k=1, ..., K$ \hspace{0.5cm} \\
\hspace{0.3cm} - \hspace{1cm} $k^* \leftarrow \arg \; \displaystyle\max_k \; \left\{ u_n(k) \right\}$ \\
\hspace{0.3cm} - \hspace{1cm} $P_n  \leftarrow p_0$ \\
\hspace{0.3cm} - \hspace{1cm} $p_n(k^*) \leftarrow P_n$ \\
\hspace{0.3cm} - \hspace{1cm} $p_n(0) \leftarrow 1-P_n$ \\
\hspace{0.3cm} - \hspace{0.3cm} end for\\\\
\hspace{0.3cm} - \hspace{0.3cm} repeat:\hspace{0.5cm} \\
\hspace{0.3cm} - \hspace{0.3cm} for $n=1, ..., N$ users do:\\
\hspace{0.3cm} - \hspace{1cm} estimate $v_n(k)$ for all $k=1, ..., K$ \\
\hspace{0.3cm} - \hspace{1cm} compute $\tilde{p}_n(k)=\max{\left\{1-\frac{e^{-1}}{v_n(k)}\;,\;0\right\}}$ \\
                                        \hspace{4cm} for all $k=1, ..., K$ \hspace{0.3cm} \\
\hspace{0.3cm} - \hspace{1cm} compute potential rates:  \\
                         \hspace{1.7cm} $\tilde{R}_n(k)=\tilde{p}_n(k)u_n(k)v_n(k)$ %for all $k\neq k^{*}$
                         \hspace{0.0cm}  \\
\hspace{0.3cm} - \hspace{1cm} if $\displaystyle\max_k \; \left\{ \tilde{R}_n(k) \right\}>\tilde{R}_n(k^{*})\left(1+\delta_R(n)\right)$ do: \\
                                    \hspace{4cm} $k^* \leftarrow \arg \; \displaystyle\max_k \; \left\{ \tilde{R}_n(k) \right\}$ \hspace{0.3cm} \\
\hspace{0.3cm} - \hspace{1cm} end if  \\
\hspace{0.3cm} - \hspace{1cm} compute $b_n(k^{*})=(1-P_n)\cdot v_n(k^{*})$ \\
\hspace{0.3cm} - \hspace{1cm} if $b_n(k^{*})>e^{-1}$ do: \\
                                        \hspace{4cm} $P_n\leftarrow P_n+\epsilon$ \hspace{0.3cm} \\
\hspace{0.3cm} - \hspace{1cm} else, do:  \\
                                        \hspace{4cm} $P_n\leftarrow P_n-\epsilon$ \hspace{0.3cm} \\
\hspace{0.3cm} - \hspace{1cm} end if  \\
\hspace{0.3cm} - \hspace{1cm} $p_n(k^*) \leftarrow P_n$ \\
\hspace{0.3cm} - \hspace{1cm} $p_n(0) \leftarrow 1-P_n$ \\
\hspace{0.3cm} - \hspace{0.3cm} end for\\
\hspace{0.3cm} - \hspace{0.3cm} until $|b_n(k)- e^{-1}|\leq \delta$ for all $n, k$\\\\
\hline
		\end{tabular}
	\label{tab:seq_algorithm}
\end{table}
\subsection{Parallel Updating}
\label{ssec:parallel}
The parallel algorithm is based on the observation that for a large number of users (and when A(1), A(2) hold) the maximal network throughput in multi-channel ALOHA networks approaches $Ke^{-1}$, where users transmit with probability $K/N$ \cite{Bai_Opportunistic_2006, Cohen_Game_2013}.
The parallel algorithm is described as follows. In the initialization step, all users set their transmission probability to $P^{(0)}_n=p_0$. In the learning step, all users monitor the channel utilization $v_n(k)$ for all $k=1, ..., K$ and compute $b_n(k)=\left(1-p_0\right)^{\hat{N}_n(k)}$. Hence, all users can estimate the number of users by:
\beq
\displaystyle \hat{N}_n=\sum_{k=1}^{K}{\hat{N}_n(k)}=\sum_{k=1}^{K}{\frac{\log(b_n(k))}{\log(1-p_0)}} \;.
\eeq
Then all users set their transmission probability:
\beq
\displaystyle P_n=\frac{K}{\hat{N}_n} \;.
\eeq
and implement the best-response dynamics, discussed in Section \ref{ssec:best}, with a given transmission probability $P_n$ . Theorem \ref{th:performance} shows that under $A(1), A(2)$, $b_n(k)\rightarrow e^{-1}$ for all $k$ as $N\rightarrow\infty$. The parallel updating mechanism is given in Table \ref{tab:par_algorithm}. \vspace{0.1cm} \\
\begin{remark}
Distributed algorithms for single-channel ALOHA networks under the fixed throughput demand $\eta_n$ of each user were proposed in \cite{Menache_Rate_2008, Jin_Equilibria_2002}. In each update step, the user sets $P_n=\eta_n/v_n$, until the algorithm converges. However, convergence is guaranteed only if the throughput demands are in the feasible region $\sum_{n=1}^{N}{\eta_n}\leq (1-1/N)^{(N-1)}$.
Hence, the parallel mechanism can be used to guarantee that the throughput demands are in the feasible region, by adjusting the throughput demands when the user population is changed randomly. \vspace{0.1cm}
\end{remark}
\begin{table}
\caption{Parallel updating algorithm}
\centering
\normalsize\begin{tabular}{|l|}
\hline
\\
\hspace{0.3cm} - \hspace{0.3cm} Initialize:\\
\hspace{0.3cm} - \hspace{0.3cm} for $n=1, ..., N$ users do:\\
\hspace{0.3cm} - \hspace{1cm} estimate $u_n(k)$ for all $k=1, ..., K$ \hspace{0.5cm} \\
\hspace{0.3cm} - \hspace{1cm} $k^* \leftarrow \arg \; \displaystyle\max_k \; \left\{ u_n(k) \right\}$ \\
\hspace{0.3cm} - \hspace{1cm} $P_n  \leftarrow p_0$ \\
\hspace{0.3cm} - \hspace{1cm} $p_n(k^*) \leftarrow P_n$ \\
\hspace{0.3cm} - \hspace{1cm} $p_n(0) \leftarrow 1-P_n$ \\
\hspace{0.3cm} - \hspace{0.3cm} end for\\\\
\hspace{0.3cm} - \hspace{0.3cm} for $n=1, ..., N$ users do:\\
\hspace{0.3cm} - \hspace{1cm} estimate $v_n(k)$ for all $k=1, ..., K$ \\
\hspace{0.3cm} - \hspace{1cm} compute $b_n(k)=(1-p_0)\cdot v_n(k)$ \\
                                        \hspace{4cm} for all $k=1, ..., K$ \hspace{0.3cm} \\
\hspace{0.3cm} - \hspace{1cm} compute $\hat{N}_n=\sum_{k=1}^{K}{\frac{\log(b_n(k))}{\log(1-P_n)}} $ \\
\hspace{0.3cm} - \hspace{0.3cm} end for\\\\
\hspace{0.3cm} - \hspace{0.3cm} for $n=1, ..., N$ users do:\\
\hspace{0.3cm} - \hspace{1cm} $P_n  \leftarrow K/\hat{N}_n$ \\
\hspace{0.3cm} - \hspace{1cm} $p_n(k^*) \leftarrow P_n$ \\
\hspace{0.3cm} - \hspace{1cm} $p_n(0) \leftarrow 1-P_n$ \\
\hspace{0.3cm} - \hspace{0.3cm} end for \\\\
\hspace{0.3cm} - \hspace{0.3cm} perform the best-response dynamics  \\
                 \hspace{1.0cm} with given $P_n$ until convergence \\\\
\hline
		\end{tabular}
	\label{tab:par_algorithm}
\end{table}

\subsection{Convergence of the Sequential and Parallel Updating Algorithms}

When applying the sequential and parallel updating algorithms, users can change their selected channels according to the dynamic load. In this section we show that the dynamics converge in finite time, starting from any point.

The following theorem establishes the convergence of the sequential updating algorithm. For purposes of analysis, we assume that users do not reduce their transmission probability to zero (thus, users with a high transmission probability should reduce their rates). Therefore, we assume that the transmission probability of every user is lower bounded by $P_n>\epsilon_p$ for some $0<\epsilon_p<<1$. \vspace{0.1cm}

%---------Theorem convergence sequential updating-------------------
\begin{theorem}
The sequential updating algorithm given in Table \ref{tab:seq_algorithm} converges to a NEP in finite time, starting from any point. \vspace{0.1cm}
\end{theorem}
%-----proof of Theorem seq---
\begin{proof}
Assume that $N-1$ users play a multi-strategy matrix $\mathbf{P}_{-n}\in \mathcal{P}_{-n}$. Assume that user $n$ has computed the potential rates $\tilde{R}_n(k)\;,\; k=1, ..., K$ and wants to update its strategy. User $n$ will switch to a different channel only if
\begin{center}
$\displaystyle\max_k \; \left\{ \tilde{R}_n(k) \right\}>\tilde{R}_n(k^{*})\left(1+\delta_R(n)\right)$
\end{center}
holds. \\
Note that $\epsilon_p\leq \tilde{p}_n(k)\leq 1-e^{-1}$ for all $n, k$. Thus,
\begin{center}
$\displaystyle\max_k \left\{ \tilde{R}_n(k) \right\}\leq\left(1-e^{-1}\right)\max_k\left\{u_n(k)\right\}$
\end{center}
and
\begin{center}
$\displaystyle\tilde{R}_n(k^{*})\geq \epsilon_p^N \min_k\left\{u_n(k)\right\}$.
\end{center}
Let
\beq
\label{eq:delta_R_n}
\delta^*_R(n)
=\frac{\left(1-e^{-1}\right)\max_k\left\{u_n(k)\right\}}{\epsilon_p^N\min_k\left\{u_n(k)\right\}}.
\eeq
Then,
\begin{center}
$\displaystyle\max_k \; \left\{ \tilde{R}_n(k) \right\}<\tilde{R}_n(k^{*})\left(1+\delta^*_R(n)\right)$.
\end{center}
As a result, user $n$ will not switch strategy in the next iterations for any multi-strategy of the other users once $\delta_R(n)>\delta^*_R(n)$ (which occurs in finite time by increasing $\delta_R(n)$ from time to time). Once $\delta_R(n)>\delta^*_R(n)$ for all $n$ occurs, the entire system is in equilibrium. \vspace{0.1cm}
\end{proof}
It should be noted that practically, simulation results show fast convergence of the sequential updating algorithm for very small values of $\delta_R(n)$.

The following theorem establishes the convergence of the parallel updating algorithm. \vspace{0.1cm}
%
%-----theorem-----------------------------
\begin{theorem}
The parallel updating algorithm given in Table \ref{tab:par_algorithm} converges to a NEP in finite time, starting from any point. \vspace{0.1cm}
\end{theorem}
%----proof----------------------------
\begin{proof}
After the initialization step, all users set their transmission probability to $P_n=K/\hat{N}_n$. Then, all the users implement the best-response dynamics discussed in Section \ref{ssec:best} with a given transmission probability $P_n$. As a result, convergence is guaranteed in finite time, starting from any point by Corollary $1$.
\end{proof}

\subsection{Achieving the Global NBS via Best Response}
\label{ssec:performance}
In this section we examine the performance of the algorithms in the asymptotic regime (i.e., as $N\rightarrow\infty$, where $K$ is fixed). For purposes of analysis, we assume that $\epsilon$ and $p_0$ can be arbitrarily small when applying the sequential updating algorithm. Theorem \ref{th:performance} shows that under assumption $A(1),A(2)$, both the sequential and parallel updating algorithms maximize the global sum log rate of the network as $N$ increases. Theorem \ref{th:global_NBS} shows that the global NBS of the network is achieved in this case.
\vspace{0.1cm}
%------------theorem---------
\begin{theorem}
\label{th:performance}
Assume that $A(1), A(2)$ hold. Then, applying the sequential and parallel updating algorithms given in Tables \ref{tab:seq_algorithm} and \ref{tab:par_algorithm} respectively, maximizes the sum log rate $\sum_{n=1}^{N}\log(R_n)$ as $N\rightarrow\infty$ with probability $1$. \vspace{0.1cm}
\end{theorem}
%-------proof----------
\begin{proof}
We prove the theorem in two steps. First, we establish the upper bound on the sum log rate that can be achieved by any algorithm. Then, we show that the proposed algorithms achieve the bound in the asymptotic regime.

We use the same notation as in the proof of Theorem \ref{th:b2}. Substituting $P_n^*=1/N(k)$ in (\ref{eq:S_k}) yields:
\beq
\label{eq:S_k_bound2}
\bea{l}
\displaystyle S_k\leq\sum_{n\in\mathcal{N}_k}\left[\log(u_n^*)\right] \vspace{0.1cm} \\ \hspace{0.5cm}-N(k)\log(N(k))+\left(N(k)-1\right)N(k)\log\left(1-\frac{1}{N(k)}\right)
\;,
\ena
\eeq
where $u_n^*=\max_k\left(u_n(k)\right)$.\\
Let $S\triangleq\sum_{k=1}^{K}S_k$ be the sum log rate of the network. Hence\footnote{The bound holds for $N(k)\geq 2$ for all $k$. It can be verified that $N(k)\leq 2$ is not a valid solution to maximize the upper bound as $N$ increases.},
\beq
\label{eq:S_bound}
\bea{l}
\displaystyle S\leq\sum_{k=1}^{K}\sum_{n\in\mathcal{N}_k}\left[\log(u_n^*)\right] \vspace{0.1cm} \\
\displaystyle+\sum_{k=1}^{K}\left[-N(k)\log(N(k)) \right. \vspace{0.1cm} \\ \left. \hspace{2cm}
         \displaystyle+\left(N(k)-1\right)N(k)\log\left(1-\frac{1}{N(k)}\right)\right] \vspace{0.1cm} \\
\displaystyle=u^*+\sum_{k=1}^{K}f\left(N(k)\right)
\;,
\ena
\eeq
where $u^*\triangleq\sum_{n=1}^{N}\left[\log(u_n^*)\right]$ is a constant independent of $n, k$ and $f\left(N(k)\right)$ is a function of $N(k)$.  \\
It can be verified that the second derivative of $f\left(N(k)\right)$ with respect to $N(k)$ is strictly negative in its domain. Therefore, by the strict concavity of $f\left(N(k)\right)$, for any partition of $N$, $N(k)=\alpha_k N$, $k=1, ..., K$, such that $\sum_{k=1}^{K}\alpha_k=1$, we have: $\sum_{k=1}^{K}\frac{1}{K}f(\alpha_k N)\leq f(\frac{1}{K}\sum_{k=1}^{K}\alpha_kN)=f(N/K)$, where equality holds iff $\alpha_k=1/K$ for all $k$.
Therefore, maximizing the upper bound with respect to $N(k)$, $k=1, ..., K$ yields a solution $N^*(k)=N/K$ for all $k$. Substituting $N^*(k)$ in (\ref{eq:S_bound}) yields:
\beq
\label{eq:S_bound2}
\bea{l}
\displaystyle S\leq u^*+N\log\left(\frac{K}{N}\right)+N\left(\frac{N}{K}-1\right)\log\left(1-\frac{K}{N}\right)
\;.
\ena
\eeq

Next, to show that the parallel algorithm achieves this bound (\ref{eq:S_bound2}), it suffices to show the following: $1)$ the users transmit with probability $P_n=K/N$ for all $n$; $2)$ every user selects the channel with the highest collision-free utility $u_n^*$; $3)$ the number of users that transmit over every channel approaches $N/K$. In what follows we show that these three requirements hold in the asymptotic regime (i.e., as $N\rightarrow\infty$ and $K$ is fixed). Note that once the users have estimated the total number of users in the network $\hat{N}_n$, they set $P_n=K/\hat{N}_n$. Assuming that each user perfectly estimates the load on all the channels, then $P_n=K/N$ for all $n$. Thus, requirement $1$ holds. In the next step, the users perform the best-response dynamics with given $P_n=K/N$ for all $n$ until convergence.
Note that in the first iteration every user $n$ selects the channel with the highest collision-free utility.
Let $k_n^*=\arg\;\max_k\left(u_n(k)\right)$ and let
\beq\label{eq:indicator_tilda}
\mathbf{\tilde{1}}_n(k)=
\left\{ \begin{matrix}
1 & , & \mbox{if $k=k_n^*$}  \\
0 & , & \mbox{otherwise}
\end{matrix} \right.  \;,
\eeq
be the indicator function, which indicates whether user $n$ tries to access channel $k$ at the first iteration. \\
Let
\beq\label{N_k_tilda}
\tilde{N}(k)=\sum_{n=1}^{N}{\mathbf{\tilde{1}}_n(k)}
\end{equation}
be the number of users that access channel $k$ at the first iteration. \\
Since $u_n(k)$ are identically distributed across channels (due to assumption $A(2)$), we have: $Pr\left(k=k_n^*\right)=1/K \;\forall k \; \forall n$. Note that $\mathbf{\tilde{1}}_n(k)$ are also independent across users (from assumption $A(1)$).
Therefore, the strong law of large numbers implies that the sample average of $\mathbf{\tilde{1}}_n(k)$ converges almost surely to the expected value ($\E\left\{\mathbf{\tilde{1}}_n(k)\right\}=1/K$). Hence,
\begin{center}
$\tilde{N}(k)\overset{a.s}{\longrightarrow} N/K$ \; $\forall k$  \; as \; $N\rightarrow\infty$ \;.
\end{center}
Thus, requirement $2, 3$ hold in the first iteration. Let $\tilde{N}_n(k)=\tilde{N}(k)-\mathbf{\tilde{1}}_n(k)$ be the number of users that access channel $k$ at the first iteration except user $n$. In the next iterations, every user observes an equal load on every channel (assuming perfect monitoring) since $v_n(k)=\left(1-K/N\right)^{\tilde{N}_n(k)}\rightarrow \left(1-K/N\right)^{N/K}=e^{-1}$ \; $\forall n, k$ \; as \; $N\rightarrow\infty$ with probability $1$. As a result, the users will not switch in the next iterations and will operate in the desired operating point with probability $1$.

A similar argument applies to the sequential updating algorithm. In the initialization step, let $P_n=p_0=\alpha/N$ for some $0<\alpha<K$ for all $n$. Then, the load on every channel approaches a constant since $v_n(k)=\left(1-\alpha/N\right)^{\tilde{N}_n(k)}\rightarrow\left(1-\alpha/N\right)^{N/k}=e^{-\alpha/K}$ \; $\forall n, k$ \; as \; $N\rightarrow\infty$ with probability $1$. Let $\Delta v(t)=\max_{n, k}(v_n(k))-\min_{n, k}(v_n(k))$ at time $t$ and set $\epsilon=\epsilon'/N$ for small $\epsilon'>0$. Let $t_1$, $t_2$ be the time indices when all the users set $p_0$, $p_0+\epsilon$, respectively (we assume that during the sequential updating every user waits a fixed amount of time between adjacent updates). Thus, $\Delta v(t)\leq e^{-\alpha/K}-e^{-(\alpha+\epsilon')/K}\leq 1-e^{-\epsilon'/K}$ for all $\alpha>0$ for all $t_1\leq t\leq t_2$ as $N\rightarrow\infty$. Thus, for any fixed $\delta_R(n)>0$ there exists $\epsilon'>0$ such that the users will not switch to other channels. As a result, the sequential updating continues until every user updates its transmission probability on the channel with the highest collision-free utility to $P_n=K/N$ for all $n$ as $N\rightarrow\infty$ with probability $1$.
\end{proof}
\hspace{0.2cm}

Next, we show that the global NBS is achieved as $N\rightarrow\infty$. Note that when a selfish user increases its transmission probability to $1$ over its best channel to increase its rate, any other user will observe a zero rate on this channel. For the next theorem we assume that users that observe zero rates on all the channels transmit over the channel with the weakest interference (which can be sensed by the transmitter or the receiver, as discussed in Section \ref{sec:network}). This assumption is reasonable from a game theoretic perspective, since it encourages selfish users to cooperate, as shown in the proof of Theorem \ref{th:global_NBS} below. It is also reasonable from a system perspective. We also assume that the interference gain $|h_{i,j}(k)|$ that user $i$ causes to user $j$ on channel $k$ is bounded by $0<|h_{min}|\leq|h_{i,j}(k)|\leq|h_{max}|<\infty$ for all $i, j$ for all $k$.
%------------theorem---------
\begin{theorem}
\label{th:global_NBS}
Assume that $A(1), A(2)$ hold. Let $\mathcal{N'}$ be the set of all the users in the network in (\ref{eq:NBS}). Then, applying the sequential and parallel updating algorithms given in Tables \ref{tab:seq_algorithm} and \ref{tab:par_algorithm} respectively, achieves the global NBS of the network as $N\rightarrow\infty$ with probability $1$. \vspace{0.1cm}
\end{theorem}
%-------proof----------
\begin{proof}
When users do not cooperate, every user transmits over the channel that yields the maximal achievable rate with a transmission probability equals to $1$. Therefore, after $K$ iterations, all the channels are occupied by $K$ users that always transmit. As a result, every user that updates its strategy at iteration $t>K$ observes a zero rate over all the channels. Then, it transmits over the channel with the weakest interference with a transmission probability equals to $1$ (to maximize the interference to the selfish users to encourage cooperation). Since $0<|h_{min}|\leq|h_{max}|<\infty$, there exists $M>0$ such that $|h_{max}|/|h_{min}|<M$. Let $N>KM$. Then, there exists a channel (say $k$) such that $N(k)>M$. Therefore, the interference $I_{k, n}$, that the users on channels $k$ cause to user $n$, is lower bounded by $I_{k,n}>M|h_{min}|>|h_{max}|$. Hence, if $M$ users transmit on channel $k$ and there is a channel which is occupied solely by a single user, in the next iteration user $n$ will not transmit on channel $k$. The same argument applies until at least two users transmit on every channel. As a result, $R_n^{min}=0$ for all $n$ (i.e., the global NBS is equivalent to maximizing the sum log rate of the network) for a sufficiently large $N$. The rest of the proof follows from Theorem \ref{th:performance}. \vspace{0.1cm}
\end{proof}

\begin{remark}
The advantages of the sequential mechanism are twofold. First, even if users start the dynamics with different transmission probabilities, they update their transmission probabilities to approach $b(k)=e^{-1}$. Second, in the case of a non-i.i.d matrix $U$, the users adjust their transmission probability according to the channel load. This property is important in common scenarios, such as when there is a significant interference on some channels, as discussed in Remark $1$. On the other hand, when users are synchronized and parallel updating can be applied, the parallel mechanism determines the required transmission probability in a single iteration. Then, convergence of the best-response dynamics with a given transmission probability is much faster. Hence, if $A(1), A(2)$ hold, this is a good solution, since it approaches the desired operating point as $N$ increases.
\end{remark}

\section{Simulation Results}\label{sec:simulations}
In this section we provide numerical examples to demonstrate the performance of the algorithms. First, we simulate the proposed best-response dynamics discussed in Section \ref{sec:rate_maximization}, for heterogenous networks, where transmission probabilities are given. We further simulate the proposed distributed algorithms discussed in Section \ref{sec:distributed}, for rate maximization, when users monitor the channel load to adjust their transmission probabilities. In all cases, the estimation of $v_n$ is based on a window of $100$ packets. We simulated Rayleigh fading channels, i.i.d across users and channels. The entries of the collision-free rate matrix $U$ were $u_n(k)=W\log(1+\mbox{SNR}\cdot |h_n(k)|^2)$ bps, where the channels' bandwidth $W$ was set to $10$MHz.
\subsection{Simulation of The Rate Maximization Under Given Transmission Probability Constraints}\label{ssec:sim_best}
In this section, we compared three algorithms: a random access algorithm where users pick a channel randomly, a totally greedy (TG) scheme where users pick the channel that maximizes their collision-free rates $u_n(k)$, and finally the proposed best-response dynamics discussed in Section \ref{ssec:best}. Transmission probabilities of the heterogenous users were uniformly distributed: $P_n\sim [0,2K/N]$ (note that $K/N$ is the desired transmission probability for rate maximization in a homogenous network). We initialized the best-response dynamics by the solution of the TG scheme. The achievable rates are presented as the ratio of the rate achieved by the random access algorithm.

In Fig. \ref{fig:rates} we present the average user rate gain of the best-response dynamics and the TG access scheme over the random access scheme as a function of the number of users for SNR$=0$dB, SNR$=10$dB, and $K=10$ channels. It can be seen that the average user rate achieved by the best-response dynamics significantly outperforms the average user rate achieved by all other algorithms. However, it approaches the TG scheme as $N$ increases, as discussed in Section \ref{sec:csi}. Note that the gain over the random access scheme decreases as the SNR increases. This is because the channel diversity gain decreases with SNR \cite{Brennan_Linear_1959}. For $N=20$ and SNR$=10$dB, the average number of iterations until convergence of the proposed best-response dynamics was $14$. \\
\begin{figure}[h]
\centering \epsfig{file=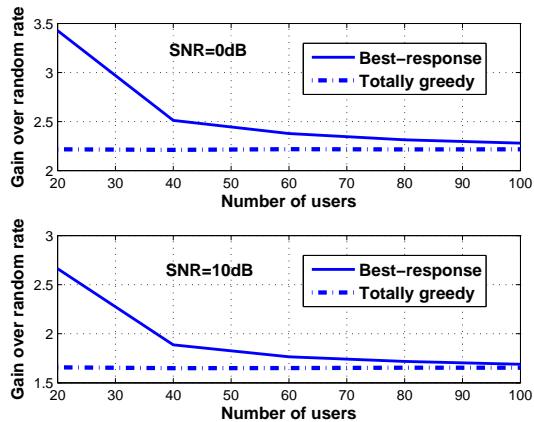,
width=0.48\textwidth}
\caption{Rate gain of the proposed best-response dynamics and the TG schemes over the random access scheme as a function of the number of users.}
\label{fig:rates}
\end{figure}

\subsection{Simulation of The Adaptive Rate Maximization}\label{ssec:sim_best}
In this section, we consider the case where users are not restricted by a transmission probability constraint, as discussed in Section \ref{sec:distributed}. Users maximize their rate, but still keep the desired load on the channels. In Fig. \ref{fig:convergence_seq} we present the convergence of the sequential updating algorithm, as shown in Table \ref{tab:seq_algorithm}, on a single channel (i.e., $K=1$) to the desired throughput $e^{-1}$. We also present the performance of the parallel scheme, given in Table \ref{tab:par_algorithm} in this case. In cases where parallel updating by all users can be implemented, this scheme is preferred on a single channel, since it only requires a single iteration.
\begin{figure}[h]
\centering \epsfig{file=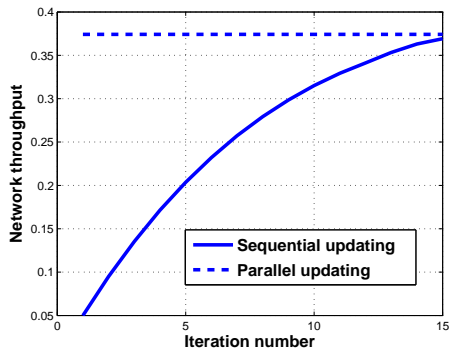,
width=0.4\textwidth}
\caption{Network throughput achieved by the sequential and parallel updating algorithms, given in Tables \ref{tab:seq_algorithm} and \ref{tab:par_algorithm}, for $N=30$.}
\label{fig:convergence_seq}
\end{figure}

Next, we illustrate the performance of the sequential updating mechanism given in Table \ref{tab:seq_algorithm}, in a multi-channel system. We simulated a common scenario where users transmit over channels $k=1, 2$ with SNR=$20dB$, and over channels $k=3, 4$ with SNR=$10dB$, due to significant interference in channels $k=3, 4$. We compare the algorithm performance for $\delta_R\rightarrow\infty$ (i.e., users transmit over the channel with the highest collision-free utility) and $\delta_R=0.1$ (i.e., users change channels only if their rates are improved by at least 10\%). We set $\delta_R$ to be equal for all users. The average rate and average log-rate as a function of the number of users are presented in Fig. \ref{fig:seq_N30_SNR_20_10_dB}. In Fig. \ref{fig:seq_N10_SNR_20_10_dB_convergence} we present the convergence of the algorithm for $N=10$ as a function of the number of iterations. In Fig. \ref{fig:seq_N10_SNR_20_10_dB_users} we present the average number of users that transmit over the inferior channels ($k=3, 4$). For $\delta_R\rightarrow\infty$, the average number of users that transmit over the inferior channels approaches zero. It can be seen in Fig. \ref{fig:seq_N30_SNR_20_10_dB} that implementing the sequential updating mechanism using $\delta_R=0.1$ (i.e., approaching the best-response dynamics) significantly outperforms the TG scheme (i.e., $\delta_R\rightarrow\infty$) in terms of both average rate (i.e., efficiency) and average log-rate (i.e., balancing between efficiency and fairness and approaching the NBS). As discussed in Section \ref{ssec:sequential} and can be seen in Fig. \ref{fig:seq_N10_SNR_20_10_dB_users}, low $\delta_R$ leads the users to use inferior channels when the load on good channels increases significantly. On the other hand, increasing $\delta_R$ leads to a high load on good channels and inefficient exploitation of the inferior channels.
\begin{figure}[h]
\centering \epsfig{file=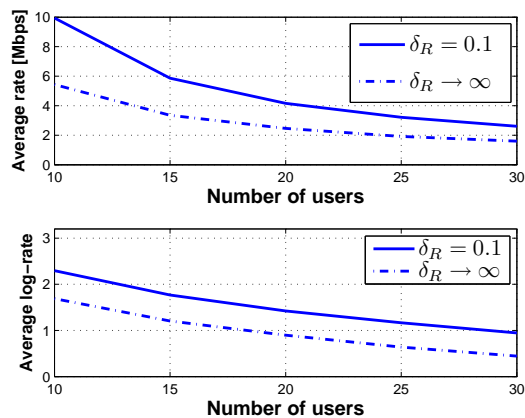,
width=0.48\textwidth}
\caption{Average user rate and log-rate achieved by the sequential updating algorithm, given in Table \ref{tab:seq_algorithm}.}
\label{fig:seq_N30_SNR_20_10_dB}
\end{figure}
\begin{figure}[h]
\centering \epsfig{file=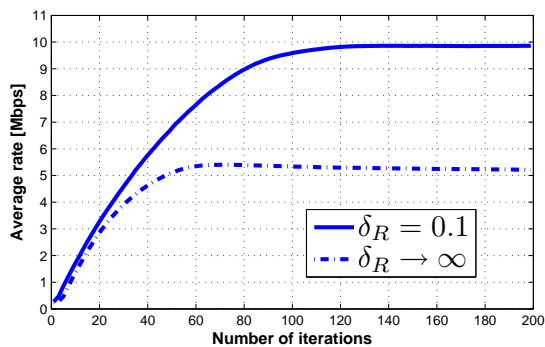,
width=0.5\textwidth}
\caption{Convergence of the sequential updating for $N=10$ as a function of the number of iterations. }
\label{fig:seq_N10_SNR_20_10_dB_convergence}
\end{figure}
\begin{figure}[h]
\centering \epsfig{file=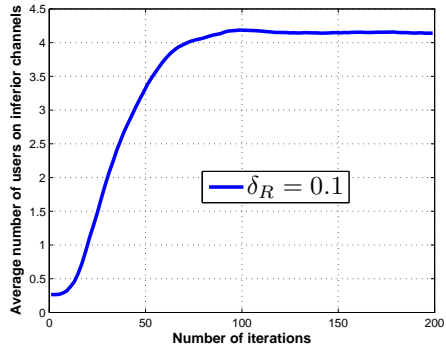,
width=0.4\textwidth}
\caption{Number of users that select the inferior channels by the sequential updating for $N=10$ as a function of the number of iterations.}
\label{fig:seq_N10_SNR_20_10_dB_users}
\end{figure}
\section{Conclusion}\label{sec:conclusion}
In this paper we examined the problem of distributed rate maximization in multi-channel ALOHA networks. We focused on networks containing a large number of users that transmit over a typically low number of channels.

First, we proposed a distributed best-response dynamics for the rate maximization problem. In this scheme, users exploit both CSI and the channel utilization to increase their rates. The convergence of the algorithm was proved for general heterogenous networks using the theory of potential games. We compared this scheme to the simple transmission scheme, where each user transmits over the channel with the highest collision-free utility.

Then, we considered the case where users can adjust their transmission probability to increase their rates. Adaptive distributed rate maximization was formulated to achieve both efficiency and fairness among users. We show that our approach plays an important role in achieving the Nash bargaining solution among users. We propose sequential and parallel algorithms to solve the optimization problem. The efficiencies of the algorithms were demonstrated through both theoretical and simulation results.

The model in this paper considered the saturated case, where users always have data to transmit. A future research direction is to examine more advanced queuing analysis for this model.
\bibliographystyle{ieeetr}
%\bibliography{Co_Le_Distributed_ToN_bib}

\end{document}

%% file: macros.tex
\newcommand{\E}{\ensuremath{\hbox{\textbf{E}}}}

\newtheorem{theorem}{Theorem}
\newtheorem{lemma}{Lemma}
\newtheorem{remark}{Remark}
\newtheorem{definition}{Definition}
\newtheorem{corollary}{Corollary}
\newtheorem{proposition}{Proposition}

\newcommand{\beq}{\begin{equation}}
\newcommand{\eeq}{\end{equation}}
\newcommand{\bea}{\begin{array}}
\newcommand{\ena}{\end{array}}
\newcommand{\bds}{\begin {itemize}}
\newcommand{\eds}{\end {itemize}}
\newcommand{\bdf}{\begin{definition}}
\newcommand{\blm}{\begin{lemma}}
\newcommand{\edf}{\end{definition}}
\newcommand{\elm}{\end{lemma}}
\newcommand{\bthm}{\begin{theorem}}
\newcommand{\ethm}{\end{theorem}}
\newcommand{\bprp}{\begin{prop}}
\newcommand{\eprp}{\end{prop}}
\newcommand{\bcl}{\begin{claim}}
\newcommand{\ecl}{\end{claim}}
\newcommand{\bcr}{\begin{coro}}
\newcommand{\ecr}{\end{coro}}
\newcommand{\bquest}{\begin{question}}
\newcommand{\equest}{\end{question}}

\newcommand{\larrow}{{\larrow}}

\def\urltilda{\kern -.15em\lower .7ex\hbox{\~{}}\kern .04em}